\documentclass[%
reprint,
 amsmath,amssymb,
 aps,
pra,
floatfix,
]{revtex4-2}

\usepackage{mathtools}
\usepackage{graphicx}
\usepackage{dcolumn}
\usepackage{bm}
\usepackage{physics}
\usepackage{tikz}
\usepackage{standalone}
\usepackage{verbatim}
\usepackage[caption=false]{subfig}
\pgfdeclarelayer{bg}
\pgfsetlayers{bg,main}

\usepackage{amsthm}
\usepackage{thmtools, thm-restate}
\usepackage{natbib}

\declaretheoremstyle[
    bodyfont=\normalfont\itshape
]{mystyle}

\declaretheorem[name=Theorem, style=mystyle]{theorem}
\declaretheorem[name=Observation, sibling=theorem, style=mystyle]{observation}
\declaretheorem[name=Corollary, sibling=theorem, style=mystyle]{corollary}
\declaretheorem[name=Lemma, sibling=theorem, style=mystyle]{lemma}

\begin{document}

\preprint{APS/123-QED}

\title{Quantum Grid States and Hybrid Graphs}

\author
{Biswash Ghimire,$^{1\ast}$ Thomas Wagner,$^{1}$ Hermann Kampermann,$^{1}$ Dagmar Bruß$^{1}$ \\
\normalsize{$^{1}$ Institute for Theoretical Physics III, Heinrich-Heine-Universität Düsseldorf, D-40225 Düsseldorf, Germany}\\
\normalsize{$^\ast$Correspondence: biswash.ghimire@hhu.de}
}

\date{\today}

\begin{abstract}
Using the signed Laplacian matrix, and weighted and hybrid graphs, we present additional ways to interpret graphs as grid states.
Hybrid graphs offer the most general interpretation.
Existing graphical methods that characterize entanglement properties of grid states are adapted to these interpretations.
These additional classes of grid states are shown to exhibit rich entanglement properties, including bound entanglement. 
Further, we introduce graphical techniques to construct bound entangled states in a modular fashion.
We also extend the grid states model to hypergraphs.
Our work, on one hand,  opens up possibilities for constructing additional families of mixed quantum states in the grid state model.
On the other hand, it can serve as an instrument for investigating entanglement problems from a graph theory perspective.

\end{abstract}

\maketitle


\section{\label{sec:level1}Introduction}

The realization that quantum entanglement can be used as a resource \cite{entanglement_as_a_resource} has garnered intense interest in the study and characterization of entanglement.
A fundamental problem is to determine whether a given quantum state is entangled or separable -- called the separability problem \cite{bruss_separability_problem}. 
It has been proven that determining whether an arbitrary quantum system is separable is an NP-hard problem \cite{Gurvits, np_hard_gharibian}. However, it can still be worthwhile to explore the problem in the context of some particular family of quantum states instead of  general states. In this paper, we focus on several families of quantum states that can be represented as combinatorial graphs, and determine entanglement properties of such states via graph theoretic methods.

Interest in interpreting so-called graph Laplacians as density matrices can be traced back to the work of Braunstein et at. \cite{Braunstein2006}, where it was shown that the normalized signed Laplacian matrix of a graph can be interpreted as a density matrix. 
This idea was refined by Lockhart et. al. in \cite{Lockhart2018, jh_phdthesis} by imposing a grid structure on graphs, called grid-labelled graphs. 
We expand on this concept and provide additional interpretations of grid-labelled graphs as quantum states, using various Laplacian matrices.

We first summarize the concept and properties of quantum grid states. Grid states, introduced  in \cite{Lockhart2018}, are mixed quantum states described by simple graphs called grid-labelled graphs. 
Note that these states are different from grid states in \cite{hastrup_measurement-free_2021}.
The vertices in a grid-labelled graph are arranged on a grid and are labelled with Cartesian indices $(i, j)$ row-wise from top-left to bottom-right. 
An edge $\{(i, j), (k, l)\}$ connecting vertices $(i, j)$ and  $(k, l)$ is interpreted as the state $1/\sqrt{2} (\ket{ij} - \ket{kl})$, called an edge state.
For example, Fig. \ref{fig:first_graph}(a) shows the vertex labelling in a grid-labelled graph with the $\ket{\phi^-} = \frac{1}{\sqrt{2}} (\ket{00}- \ket{11})$ Bell state. 
With this convention, the density matrix $\rho(G)$ of a grid state is defined as the equally weighted mixture of all projectors onto edge states in the corresponding grid-labelled graph $G$.

The (signed) Laplacian matrix of a grid-labelled graph, with a suitable normalization, is identical to its corresponding density matrix.
In order to see this, remember that the signed Laplacian matrix $L$ of a graph on $n$ vertices is the $n \times n$ matrix defined as
\begin{equation}
L = D - A,
\end{equation}
where $D$ is the degree matrix and $A$  the adjacency matrix \cite{diestel2010}. 
The degree matrix $D$ is an $n \times n$ diagonal matrix, in which each diagonal entry $D_{\alpha \alpha}$, where $1 \leq \alpha \leq n$, indicates the number of edges connecting to vertex $v_\alpha$ -- called the degree of vertex $v_\alpha$.  
The adjacency matrix $A$ is an $n \times n$ binary matrix such that if vertices $v_\alpha$ and $v_\beta$  are connected by an edge, the matrix entry $A_{\alpha \beta}$ is $1$, otherwise it is $0$ \cite{diestel2010}.
\begin{figure}
\subfloat[\label{subfig:phi_plus_l_graph}]{
    
       
\begin{tikzpicture}[yscale=-1]
       \foreach \x in {0,...,{1}}
       {
            \foreach \y in {0,1}
            {
            
                \node[circle, draw=black, line width = 2pt, scale=0.8,] (\y\x) at (1.4*\x, 1.4*\y) {(\y,\x)};
                
            }
       }
       \draw[solid,line width = 2pt] (00) -- (11);
   \end{tikzpicture}
       
}
\hfill
\subfloat[\label{subfig:cross_hatch}]{
    
       
\begin{tikzpicture}[yscale=-1]
           \foreach \x in {1,...,{3}}
           {
            \foreach \y in {1,...,{3}}
            {
                \node[circle, draw=black,line width = 1.5pt, scale=0.8, fill=black] (\y\x) at (0.8*\x, 0.8*\y) {};
            }
           }
           
           \foreach \v/\w in {11/23, 12/31, 21/33, 32/13}
           {
            \draw[-,line width = 1.5pt] (\v) -- (\w);
           }
\end{tikzpicture}
       
}
\hfill
\subfloat[\label{subfig:cross_hatch_transpose}]{
    
       
\begin{tikzpicture}[yscale=-1]
           \foreach \x in {1,...,{3}}
           {
            \foreach \y in {1,...,{3}}
            {
                \node[circle, draw=black,line width = 1.5pt, scale=0.8, fill=black] (\y\x) at (0.8*\x, 0.8*\y) {};
            }
           }
           
           \foreach \v/\w in {21/13, 32/11, 31/23, 12/33}
           {
            \draw[-,line width = 1.5pt] (\v) -- (\w);
           }
\end{tikzpicture}
       
}
\caption{(a) $L$-graph of the $\ket{\phi^-}$ Bell state.
The pairs of integers indicate vertex indices. 
(b)  A $3 \times 3$ cross-hatch graph, and (c) its partial transpose.
}
\label{fig:first_graph}
\end{figure}
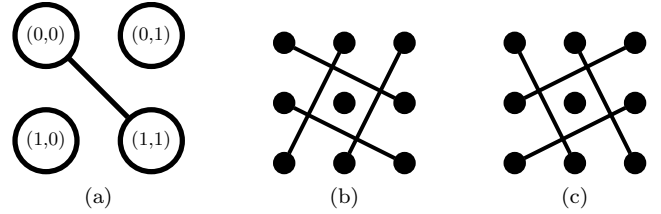

We call the grid-labelled graphs from \cite{Lockhart2018} $L$-graphs.
The degree criterion \cite{Braunstein2006, Lockhart2018} and the graph surgery procedure \cite{Lockhart2018} characterize  entanglement properties of grid states corresponding to $L$-graphs.
The degree criterion is a graphical method that can be used to verify if the density matrix of an $L$-graph is positive under partial transpose.
It makes use of the concept of partial transpose of a graph.
The partial transpose of an $L$-graph $G$ is another $L$-graph $G^\Gamma$ such that an edge
\(
\{(i,l), (k,j)\}
\)
exists in $G^\Gamma$ if and only if the edge
\(
\{(i,j), (k,l)\}
\)
exists in $G$.

\begin{restatable}[Degree Criterion for $L$-graphs from \cite{Braunstein2006, Lockhart2018}]{theorem}{degree-criterion-for-lgraphs}
   The density matrix $\rho(G)$ of an $L$-graph $G$ is positive under partial transpose if and only if $D(G) = D(G^\Gamma)$.
\end{restatable}

For example, the cross-hatch graph from \cite{Lockhart2018}, shown in Fig. \ref{fig:first_graph}(b), satisfies $D(G) = D(G^\Gamma)$.
The corresponding density matrix is therefore positive under partial transpose.

The graph surgery procedure \cite{Lockhart2018} is a graphical method that allows to verify entanglement using the range criterion \cite{Horodecki1997}.
We restate the corollary of the range criterion from \cite{Lockhart2018} as it also is the basis for graph surgery procedures presented in this paper.

\begin{restatable}[\cite{Lockhart2018}]{corollary}{range-criterion-corollary}
\label{coro:range_criterion_corollary}
If a rank $r$ density matrix has less than $r$ product vectors in its range, then it is entangled.
\end{restatable}
Graph surgery involves performing a sequence of row and column surgeries on an $L$-graph.
Row surgery is carried out by first selecting an isolated vertex, say
\(
(i,j),
\)
in the $L$-graph and performing the ``CUT'' step, in which all edges connected to vertices in row $i$ are removed. 
This is followed by the ``STITCH" step, which reconnects the path between every pair of vertices not in row $i$, if the ``CUT'' step severed the path \cite{Lockhart2018}.
In column surgery, the ``CUT'' and the ``STITCH'' steps are performed on the vertices in column $j$.
The graph produced by a row / column surgery is denoted as $G_{ij}^R$ / $G_{ij}^C$, where the superscript indicates the type of surgery -- $R$ for row surgery and $C$ for column surgery, and the subscript $ij$ denotes the isolated vertex chosen for the surgery. 
In effect, row / column surgery produces a simpler graph with fewer edges, unless vertices in the target row / column are all isolated vertices.
Figure \ref{fig:L_graph_surgery_example} shows an example of a row and a column surgery on an $L$-graph.

It was shown in \cite{Lockhart2018} that any product vector in the range of the density matrix $\rho(G)$ of an $L$-graph $G$ -- and thereby in the range of $L(G)$ --  must also be in the range of either $L(G_{ij}^R)$ or of $L(G_{ij}^C)$.
Since $G_{ij}^R$ and $G_{ij}^C$  are $L$-graphs, further row / column surgeries can be performed on them, and on the resulting graphs, and so on.
Therefore, if iterated graph surgery on an $L$-graph $G$ always leads to the empty graph $G_E$, then any product vector in the range of $L(G)$ must also be in the range of $L(G_E)$, which is the zero matrix.
This is clearly a contradiction, which means there are no product vectors in the range of $L(G)$. And, the corresponding density matrix is entangled according to Corollary \ref{coro:range_criterion_corollary}.

\begin{figure}
\subfloat[\label{subfig:example_graph} $G$]{

       
\begin{tikzpicture}[yscale=-1]
           \foreach \x in {1,...,{3}}
           {
            \foreach \y in {1,...,{2}}
            {
                \node[circle, draw=black,line width = 1.5pt, scale=0.8, fill=black] (\y\x) at (0.8*\x, 0.8*\y) {};
            }
           }
           
           \foreach \v/\w in {11/22, 22/13}
           {
            \draw[-,line width = 1.5pt] (\v) -- (\w);
           }
           \node[circle, draw=orange,line width = 1.5pt, scale=0.8, fill=orange] (Z) at (0.8*1, 0.8*2) {};
\end{tikzpicture}
       
}
\hfill
\subfloat[\label{subfig:example_graph_row_surgery} $G^R_{10}$]{

       
\begin{tikzpicture}[yscale=-1]
\foreach \x in {1,...,{3}}
{
    \foreach \y in {1,...,{2}}
    {
        \node[circle, draw=black, line width= 1.5pt, scale=0.8, fill=black] (\y\x) at (0.8*\x, 0.8*\y) {};
    }
}
    \node[circle, fill=black, scale=0.8] at (0.8*2,0.8*2) {};

\draw[solid, line width=1.5pt] (11) -- (1.6, 1.18) -- (13);
\node[circle, draw=orange,line width = 1.5pt, scale=0.8, fill=orange] (Z) at (0.8*1, 0.8*2) {};

\end{tikzpicture}
       
}
\hfill
\subfloat[\label{subfig:example_graph_col_surgery} $G^C_{10}$]{

       
\begin{tikzpicture}[yscale=-1]
           \foreach \x in {1,...,{3}}
           {
            \foreach \y in {1,...,{2}}
            {
                \node[circle, draw=black,line width = 1.5pt, scale=0.8, fill=black] (\y\x) at (0.8*\x, 0.8*\y) {};
            }
           }
           
           \foreach \v/\w in {22/13}
           {
            \draw[-,line width = 1.5pt] (\v) -- (\w);
           }
           \node[circle, draw=orange,line width = 1.5pt, scale=0.8, fill=orange] (Z) at (0.8*1, 0.8*2) {};
\end{tikzpicture}
       
}
\caption{(a)  An $L$-graph $G$.
(b) and (c) Graphs $G^R_{10}$ and $G^C_{10}$ produced respectively by row and column surgeries on $G$ with vertex $(1,0)$ in orange as the selected isolated vertex. 
For row surgery, all edges connected to vertex $(1,1)$ are removed in the CUT step.
As a result, the vertices $(0,0)$ and $(0, 2)$ get disconnected, and then get reconnected in the STITCH step, which produces graph $G^R_{10}$ in (b).
Likewise, for column surgery, all edges connected to vertex $(0, 0)$ are removed. 
This does not disconnect any path between vertices not in column 0. 
The STITCH step is therefore not necessary.
The graph $G^C_{10}$ in (c) is the result.
}
\label{fig:L_graph_surgery_example}
\end{figure}
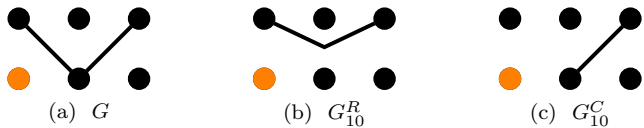

The degree criterion and the graph surgery procedure connect entanglement properties of grid states to structural properties of $L$-graphs.
Together, they enable the construction of bound entangled grid states.
Furthermore, genuine multipartite entanglement is found in higher-dimensional grid states \cite{Lockhart2018}.
Such rich entanglement properties raise further questions. 
Are there other ways to interpret grid-labelled graphs as quantum states?
Would the states also exhibit entanglement properties such as bound entanglement? 
Can the degree criterion and the graph surgery procedure be extended to such new interpretations of grid-labelled graphs?
In this paper, we investigate these questions using additional types of Laplacian matrices. 
Specifically, the notion of grid-states is extended using the signless Laplacian matrix, and the weighted signed and signless Laplacian matrices.
These interpretations lead us to conceive hybrid graphs, which represent density matrices that are mixtures of edge states corresponding to the signed and the signless Laplacian matrices.
For these new states, we derive the analogous degree criteria and graph surgery procedures, and use them to construct bound entangled states.
As a proof of concept, we also show that a degree criterion can be derived for grid-labelled hypergraphs.

We largely follow the nomenclature from  \cite{Lockhart2018}.
For clarity, we occasionally prefix certain terms with the letter symbols of corresponding Laplacian matrices.
For example, we call the grid-labelled graphs from  \cite{Lockhart2018} $L$-graphs.
Further, we make no distinction between Laplacian matrices and density matrices when normalization is irrelevant.
Similarly, since only bipartite quantum systems are considered in this paper, the partial transpose of a matrix $M$ is denoted by $M^\Gamma$ without loss of generality, as it is only used in relation to the Peres-Horodecki (also PPT) criterion \cite{Horodecki1996}, which does not depend on the transposed subsystem. 
We write a graph as $G = (V, E)$, where $V$ and $E$ are the vertex and the edge sets.
Throughout this paper, we always assume that the density matrices are normalized.
Additionally, depending on the context, we may use both a boldface letter or the braket notation for representing vectors.
For example, for product vectors, the braket notation is the clearer notation.

With the following observation it is possible to check if the degree criterion can be adapted to a new interpretation of grid-labelled graphs.

\begin{restatable}{observation}{degreecriterionquickcheck}
    \label{obs:applicability_of_degree_criterion_quick_check}
    Let $G$ be a grid-labelled graph on $n$ vertices and $\rho(G)$ be the corresponding density matrix via any of the interpretations mentioned previously. If a vector $\mathbf{v}$ with all its components equal to $\pm 1$ (henceforth $\mathbf{v}  \in \{-1, 1\}^n$) exists in the kernel of $\rho(G^\Gamma)$, and  if $\rho(G)$ is separable, then $D(G) = D(G^\Gamma)$.
\end{restatable}

The observation is proven in Appendix \ref{apdx:additional_graph_theory_concepts}.

\section{\label{sec:signless_Laplacians}$Q$-Grid States}
In this section, grid-labelled graphs are interpreted with the signless Laplacian matrix.
The signless Laplacian of a graph $G$  is defined as \( Q = D + A \), where $D$ and $A$ are the degree and the adjacency matrices of $G$. 
Normalized, the signless Laplacian is a proper density matrix.
We call the quantum states described by the normalized signless Laplacian $Q$-grid states.
The corresponding graphs are called $Q$-graphs. 
Graph features such as grid structure and vertex labelling are unchanged for $Q$-graphs, while the interpretation of edges $\{(i,j), (k,l)\}$ changes. 
A $Q$-edge state has the form
\(1/\sqrt{2} (\ket{ij} + \ket{kl})\).
The density matrix of a $Q$-grid state represented by a $Q$-graph $G = (V, E)$ is defined as
\begin{equation}
\label{eqn:q_graphs_density_matrix_definition}
\rho_Q(G) = \frac{1}{\abs{E}}\sum_{e \in E} \dyad{e} = \frac{1}{\abs{E}} Q(G),
\end{equation}
where $\{\ket{e}\}$ are the $Q$-edges states of edges in $E$.
The notion of partial transpose of $L$-graphs in  \cite{Lockhart2018}  is directly applicable to $Q$-graphs because it does not depend on the sign of the Laplacian matrix.

\begin{figure}
    \centering
    \hspace*{\fill}
    \subfloat[$G_1$]
    {
        \begin{tikzpicture}[yscale=-1]
\foreach \x in {0,...,{1}}
       {
            \foreach \y in {0,1}
            {
                \node[circle, draw=black, fill=black,line width = 1.5pt, scale=0.8 ] (\y\x) at (1*\x, 1*\y) {};
            }
       }
       \foreach \v/\w in {00/11, 00/10,10/11}
       {
        \draw[-,line width = 1.5pt] (\v) -- (\w);
       }
\end{tikzpicture}
    } 
    \hfill
    \subfloat[$G_2$]
    {
        \begin{tikzpicture}[yscale=-1]
\foreach \x in {0,...,{1}}
       {
            \foreach \y in {0,1}
            {
                \node[circle, draw=black, fill=black,line width = 1.5pt, scale=0.8 ] (\y\x) at (1*\x, 1*\y) {};
            }
       }
       \foreach \v/\w in {01/10, 00/10,10/11}
       {
        \draw[-,line width = 1.5pt] (\v) -- (\w);
       }
\end{tikzpicture}
    }
    \hspace*{\fill}
    \caption{
    (a) Graph $G_1$. (b) Graph $G_2$.
    Graphs $G_1$ and $G_2$ are partial transposes of each other.
    }
    \label{fig:q_l_surgery_comparison}
\end{figure}
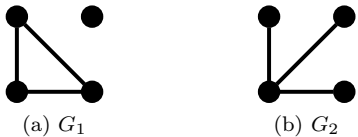

In the following, we adapt the degree criterion and the graph surgery procedure to $Q$-graphs. 
We use Observation \ref{obs:applicability_of_degree_criterion_quick_check} to identify $Q$-graphs for which the degree criterion is applicable.
The observation requires that for a $Q$-graph $G$ on $n$ vertices the signless Laplacian $Q(G^\Gamma)$ of its partial transpose graph must have a vector $\mathbf{v} \in \{-1,1\}^n$ in its kernel. 
This is only fulfilled for bipartite graphs (see Lemma \ref{lem:singless_Laplacian_null_space_only_if_bipartite}).
Therefore, we require this condition on the partial transpose of the graph.
Remember that a graph is bipartite if its vertex set can be divided into two disjoint subsets such that no edge in the graph connects vertices in the same subset.
%
%
%
%
%
%
%
%
\begin{restatable}[Degree Criterion for $Q$-graphs]{theorem}{degreecriterionqgraphs}
\label{thm:degree_criterion_Q}
Let $G$ be a $Q$-graph. If $\rho_Q(G)$ is separable and $G^\Gamma$ is bipartite, then $ D(G) = D(G^{\Gamma})$.
\end{restatable}

The proof of Theorem \ref{thm:degree_criterion_Q} is found in Appendix \ref{apdx:signless_Laplacian}. The degree criterion for $Q$-graphs, like its counterpart for $L$-graphs, is necessary and sufficient for $2 \times 2$ and $2 \times 3$ systems, due to the PPT criterion.
The bipartite condition for the graph transpose in the degree criterion for $Q$-graphs has an important  implication. 
There exist grid-labelled graphs that, if interpreted as $Q$-graphs, are separable, but are entangled if interpreted as $L$-graphs.
For example, the graphs $G_1$ and $G_2$ in Fig. \ref{fig:q_l_surgery_comparison}, if treated as $L$-graphs, represent entangled states because $D(G_1) \neq D(G_2)$.
If instead both are treated as $Q$-graphs, $G_1$ still represents an entangled state because $G_2$ is bipartite and $D(G_1) \neq D(G_2)$.
On the other hand, the degree criterion is not applicable to $G_2$ because its partial transpose $G_1$ is not bipartite.
It is easily verified that the density matrix $\rho_Q(G_2)$ is separable.

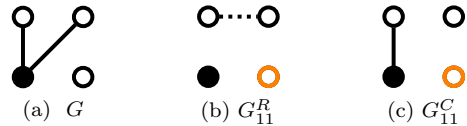
\begin{figure}
\centering
\hspace*{\fill}
\subfloat[\label{subfig:graph_components} $G$]{

    \begin{tikzpicture}[yscale=-1]
   \foreach \x in {0,...,{1}}
   {
        \foreach \y in {0,1}
        {
            \node[circle, draw=black,line width = 1.5pt, scale=0.8 ] (\y\x) at (0.8*\x, 0.8*\y) {};
        }
   }
   \foreach \v/\w in {00/10, 01/10}
   {
    \draw[-,line width = 1.5pt] (\v) -- (\w);
   }

   \node[circle, draw=black, fill=black, line width = 0.8pt, scale=0.8] at (0.8*0, 0.8*1){};
 
   \end{tikzpicture}
   
}
\hfill
\subfloat[$G_{11}^R$]{
    \begin{tikzpicture}[yscale=-1]
   \foreach \x in {0,...,{1}}
   {
        \foreach \y in {0,1}
        {
            \node[circle, draw=black,line width = 1.5pt, scale=0.8 ] (\y\x) at (0.8*\x, 0.8*\y) {};
        }
   }
   \foreach \v/\w in {00/01}
   {
    \draw[dotted,line width = 1.5pt] (\v) -- (\w);
   }

   \node[circle, draw=orange, line width = 1.6pt, scale=0.8 ] at (0.8*1, 0.8*1){};
   \node[circle, draw=black, fill=black, line width = 0.8pt, scale=0.8 ] at (0.8*0, 0.8*1){};
   
   \end{tikzpicture}
   
}
\hfill
\subfloat[$G_{11}^C$]{
    \begin{tikzpicture}[yscale=-1]
   \foreach \x in {0,...,{1}}
   {
        \foreach \y in {0,1}
        {
            \node[circle, draw=black,line width = 1.5pt, scale=0.8] (\y\x) at (0.8*\x, 0.8*\y) {};
        }
   }
   \foreach \v/\w in {00/10}
   {
    \draw[-,line width = 1.5pt] (\v) -- (\w);
   }
   
   \node[circle, draw=orange, line width = 1.6pt, scale=0.8 ] at (0.8*1, 0.8*1){};
   \node[circle, draw=black, fill=black, line width = 0.8pt, scale=0.8 ] at (0.8*0, 0.8*1){};
   
   \end{tikzpicture}
   
}
\hspace*{\fill}
\caption{Vertices are colored in black and white to show that the graphs are bipartite.
Then vertex chosen for graph surgery is indicated in orange.
Solid and dashed edges indicate $Q$- and $L$-edges, respectively. (a) $Q$-graph $G$. 
%
%
(b) Graph $G^R_{11}$. 
The CUT step splits the connected component with 
vertices $(0,0)$, $(0,1)$, and $(1,0)$.
Since vertices $(0,0)$ and $(0,1)$ are in the same partition, it is not possible to reconnect them  with a $Q$-edge. So an $L$-edge is used. 
(c) Graph $G^C_{11}$.}
\label{fig:Q-graph-example}
\end{figure}

We now extend the graph surgery procedure to $Q$-graphs. 
We call graph surgery on $Q$-graphs $Q$-surgery. 
To understand $Q$-surgery, we need the concept of connected components. 
A connected component of a graph is a subgraph that has a path between any two of its vertices, and no paths between any of its vertices and  the remaining vertices of the original graph.
An isolated vertex trivially satisfies the definition and is considered a connected component.
For example, the graph $G_1$ in Fig. \ref{fig:q_l_surgery_comparison} has two connected components and the graph $G_2$ has one.
Like $L$-surgery, $Q$-surgery is a sequence of row and/or column surgeries. 
For simplicity, $Q$-surgery is only defined for bipartite $Q$-graphs. Row surgery is performed as follows:
\begin{itemize} 

    \item CUT: Select an isolated vertex $(i,j)$ and remove all edges attached to vertices in row $i$.

    \item STITCH: If the CUT step splits any connected component and vertices in the split constituents, excluding the ones in row $i$, all belong to the same partition, reconnect the constituents with $L$-edge(s).
    Otherwise, reconnect the constituents with  $Q$-edge(s).
\end{itemize}
Note that in the STICH step, if $Q$-edge(s) are used for reconnection, each $Q$-edge must connect vertices in opposing partitions.

Likewise, column surgery is performed on vertices in column $j$.
The graph resulting from a row/column surgery on vertex $(i, j)$ is denoted as $G^R_{ij}$ / $G^C_{ij}$.

An iteration of row/column surgery on an $L$-graph always produces an $L$-graph. 
In contrast, the analogous case is not necessarily true for $Q$-graphs.
Suppose a connected component of a $Q$-graph split in the CUT step is reconnected in the STITCH step with $Q$-edge(s), while another split connected component is reconnected with $L$-edge(s).
The resulting graph is then not a $Q$-graph because it has both $L$- and $Q$-edges in it.
Nonetheless, it still holds for $Q$-graphs that any product vector in the range of the density matrix of the original $Q$-graph must be in the range of the density matrix of the graph produced after an iteration of a row and column surgery.
This is formalized in the following observation.

\begin{restatable}{observation}{graphsurgeryqgraphs}
\label{obs:graph_surgery_q_graphs}
Let $G$ be a bipartite $Q$-graph on $n$ vertices with an isolated vertex $(i, j)$.
If a product vector $\ket{\mu \, \nu} \in  R[\rho_Q(G)]$, where $R$ denotes the range, then 
\paragraph*{\textbullet} 
$\ket{\mu \, \nu} \in R[\rho_Q(G^R_{ij})]$ or  $ R\big[\rho_L(G^R_{ij})\big]$, or
\paragraph*{\textbullet}
$\ket{\mu \, \mu} \in R\left[\rho_Q(G^C_{ij})\right]$ or  $ R\left[\rho_L(G^C_{ij})\right]$, or
\paragraph*{\textbullet}
$\ket{\mu \, \nu} \in R\left[\rho(G')\right]$,

where $G'$ is a hybrid graph (see Section \ref{sec:hybrid_Laplacian}).
\end{restatable}

The proof of Observation \ref{obs:graph_surgery_q_graphs} is found in Appendix \ref{apdx:signless_Laplacian}.
An example of row and column surgeries on a $Q$-graph is shown in Fig. \ref{fig:Q-graph-example}.
With Observation \ref{obs:graph_surgery_q_graphs}, $Q$-surgery, like $L$-surgery, can be used in connection with Corollary \ref{coro:range_criterion_corollary}.
Therefore, if $Q$-surgery on a $Q$-graph always produces the empty graph, the associated density matrix is entangled.

In general, the $Q$- and $L$-grid states of the same grid-labelled graph are not unitarily equivalent.
In the following observation, we identify a condition when that is the case.

\begin{restatable}{observation}{qlgraphsunitarytransformation}
\label{obs:Q_L_graphs_unitary_transformation}
Let $G$ be a grid-labelled graph.
If $G$ is not bipartite, then $\rho_L(G)$ and $\rho_Q(G)$ are not unitarily equivalent. 
\end{restatable}

A proof of Observation \ref{obs:Q_L_graphs_unitary_transformation} is given in Appendix \ref{apdx:signless_Laplacian}.

\section{\label{sec:weighted_Laplacians} Grid states corresponding to weighted Graphs}
Weighted graphs generalize the notion of edges in graphs and allow non-zero, positive weights to be associated with each edge in the graph \cite{Mohar1997}.
In this section, the weighted signed and signless Laplacian matrices are interpreted as quantum states that correspond to the respective weighted $L$- and $Q$-graphs. 

Edge states in a weighted $L$- or a $Q$-graph have the same form as in their unweighted counterparts.
However, the density matrix is defined as
\begin{equation}
\label{eqn:weigted_graph_rho_definition}
\rho(G_w) = \frac{1}{\sum_e w_e} \sum_{e \in E} w_e\dyad{e}
\end{equation}
where $G_w$ is a weighted grid-labelled graph, $\{\ket{e}\}$ are the edge states of edges in $G_w$, and $\{w_e\}$ the respective non-zero, positive edge weights.
If the edges denote $L$-edge states ($Q$-edge states), the density matrix is the normalized signed (signless) Laplacian of the weighted graph.
The signed and the signless Laplacian matrices of weighted graphs are defined as $L = D - A$ and $Q = D + A$, respectively. 
The degree of a vertex in a weighted graph is the sum of edge weights of all edges that connect to it, and the degree matrix $D$ is a diagonal matrix with degrees of vertices as its diagonal entries.
Likewise, the adjacency matrix $A$ also accounts for edge weights.
The matrix entry $A_{\alpha \beta}$ is $w_{\alpha \beta}$ if vertices $v_\alpha$ and $v_\beta$ are connected by an edge weighted $w_{\alpha \beta}$, otherwise it is 0 \cite{Mohar1997}.
Notice that in an unweighted graph all edge weights are implicitly 1.

The edges in the partial transpose graph $G^\Gamma$ of a weighted grid-labelled graph $G$ carry the weights of the corresponding edges in $G$.
 The degree criteria and the graph surgery procedures on unweighted $L$- and $Q$-graphs directly apply to weighted graphs. 
 Lemma
\ref{lem:null_space_of_weighted_unweighted_graph_Laplacian_conincice}
justifies this claim.

\begin{restatable}{lemma}{nullspaceofweightedunweightedgraphlaplacianconincice}
\label{lem:null_space_of_weighted_unweighted_graph_Laplacian_conincice}
If the vertex and the edge sets of two weighted $L$-graphs (resp. $Q$-graphs) are identical, their  signed (resp. signless) Laplacians have identical kernels.
\end{restatable}

The proof of Lemma \ref{lem:null_space_of_weighted_unweighted_graph_Laplacian_conincice} is found in Appendix \ref{apdx:weighted_graphs}. 
With Lemma \ref{lem:null_space_of_weighted_unweighted_graph_Laplacian_conincice} and Observation \ref{obs:applicability_of_degree_criterion_quick_check},
the degree criteria for unweighted $L$- and $Q$-graphs are also valid for weighted $L$- and $Q$-graphs.
Likewise, $L$- and $Q$ surgeries also directly apply to weighted graphs.
Since Laplacian matrices are hermitian, Lemma \ref{lem:null_space_of_weighted_unweighted_graph_Laplacian_conincice} implies that Laplacians of weighted graphs with identical vertex and edge sets have identical ranges. 
This means if graph surgery on an unweighted $L$- or $Q$-graph always yields the empty graph, it must be that graph surgery on any other weighted graph with the same vertex and edge sets must also always yield the empty graph.
Therefore, edge weights are irrelevant for graph surgery and the graph surgery procedures for unweighted $L$- and $Q$-graphs can be used on weighted $L$- and $Q$-graphs.
Edge weights alone also do not determine if the density matrix corresponding to a weighted $L$- or $Q$-graph is entangled or separable. 

Moreover, Observation \ref{obs:Q_L_graphs_unitary_transformation} can be applied to weighted $Q$-graphs as formalized in the following corollary.

\begin{restatable}{corollary}{weightedqlgraphsunitarytransformation}
\label{coro:weighted_Q_L_graphs_unitary_transformation}
Let $G_w$ be a weighted grid-labelled graph. 
If $G_w$ is not bipartite, then $\rho_L(G_w)$ and $\rho_Q(G_w)$ are not unitarily equivalent.
\end{restatable}
The corollary is proved in Appendix \ref{apdx:weighted_graphs}.

\section{\label{sec:hybrid_Laplacian} Grid states with hybrid graphs}
 In this section, we approach the idea of interpreting graphs as quantum states from a physical point of view.
A density matrix that is a mixture of both $L$- and $Q$-edge states is not unphysical.
Is it then possible to represent such density matrices using grid-labelled graphs?
We answer this question in the affirmative by introducing the notion of hybrid graphs and describing analogous degree criteria and graph surgery procedures for them.

A hybrid graph contains both $L$- and $Q$-edges and is written as $G = (V, E_L + E_Q)$, where $V$ is the vertex set, and $E_L$ and $E_Q$ are the sets of $L$- and $Q$-edges, respectively. 
Its $L$- and $Q$-subgraphs are the graphs $S_l = (V, E_L)$ and $S_q = (V, E_Q)$.
Hybrid graphs slightly resemble signed graphs \cite{signed_graphs}, where each edge in a graph is given either a positive or a negative sign. 
However, we do not use the Laplacian matrix in \cite{signed_graphs} to derive the density matrix of hybrid graphs.
Instead, we treat hybrid graphs as compositions of $L$ and $Q$-graphs and define the hybrid Laplacian matrix as
\(
\mathcal{L}(G) = L(S_l) + Q(S_q).
\)
The normalized hybrid Laplacian is a density matrix that is the equally weighted mixture of all $L$- and $Q$-edge states in the corresponding graph.

Coexistence of $L$- and $Q$-edges limit general results on entanglement properties, because Observation \ref{obs:applicability_of_degree_criterion_quick_check} imposes different conditions on $L$- and $Q$-graphs.
Considering that, hybrid graphs are divided into three categories based on their edge-vertex characteristics:
\begin{itemize}
    \item{Non-Overlapping Incidence (NOI):}
    A hybrid graph with NOI has a bipartite $Q$-subgraph  and no vertex in it is connected by both a $Q$-edge and an $L$-edge. 
    \item{Conditionally Overlapping Incidence (COI):}
    A hybrid graph with COI has a bipartite $Q$-subgraph and every $L$-edge in it connects vertices that are both in the same partition.
    \item{General Incidence (GI):} 
    Hybrid graphs with GI have no restrictions on incidences of $L$- and $Q$-edges.
\end{itemize}

We call a hybrid graph with NOI a NOI-graph, and likewise for graphs with COI and GI. 
An example each of a NOI-, a COI-, and a GI-graph is given in Fig. \ref{fig:composite-graph-examples}.
Note that a NOI-graph is a special case of a COI graph, because vertices connected by $L$-edges in a NOI-graph can all be put in one of the two vertex partitions.

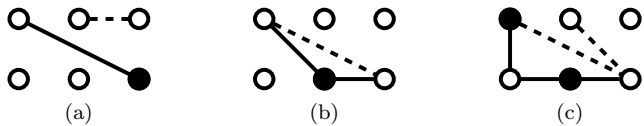
\begin{figure}
\subfloat[]{

    \begin{tikzpicture}[yscale=-1]
   \foreach \x in {1,...,{3}}
   {
    \foreach \y in {1,...,{2}}
    {
        \node[circle, draw=black,line width = 1.5pt, scale=0.8] (\y\x) at (0.8*\x, 0.8*\y) {};
    }
   }
   \foreach \v/\w in {12/13}
   {
    \draw[dashed,line width = 1.5pt] (\v) -- (\w);
   }
   \foreach \v/\w in {11/23}
   {
    \draw[-,line width = 1.5pt] (\v) -- (\w);
   }
   
   \node[circle, fill] at (0.8*3,0.8*2) {};
\end{tikzpicture}
}
\hfill
\subfloat[]{
    \begin{tikzpicture}[yscale=-1]
   \foreach \x in {1,...,{3}}
   {
    \foreach \y in {1,...,{2}}
    {
        \node[circle, draw=black, line width = 1.5pt, scale=0.8] (\y\x) at (0.8*\x, 0.8*\y) {};
    }
   }
       \node[circle, fill=black] at (0.8*2, 0.8*2) {};
       \foreach \v/\w in {11/22, 22/23}
   {
    \draw[-,line width = 1.5pt] (\v) -- (\w);
   }
   \foreach \v/\w in {11/23}
   {
    \draw[dashed,line width = 1.5pt] (\v) -- (\w);
   }
\end{tikzpicture}
}
\hfill
\subfloat[]{
    \begin{tikzpicture}[yscale=-1]
   \foreach \x in {1,...,{3}}
   {
    \foreach \y in {1,...,{2}}
    {
        \node[circle, draw=black, line width=1.5pt, scale=0.8] (\y\x) at (0.8*\x, 0.8*\y) {};
    }
   }
   \node[circle, fill=black] at (0.8*1,0.8*1) {};
   \node[circle, fill=black] at (0.8*2,0.8*2) {};
   
   \foreach \v/\w in {11/21,21/22,22/23}
   {
    \draw[-,line width = 1.5pt] (\v) -- (\w);
   }
   \foreach \v/\w in {12/23,11/23}
   {
    \draw[dashed,line width = 1.5pt] (\v) -- (\w);
   }
\end{tikzpicture}
}
\caption{Three types of hybrid graphs. 
Vertices are colored in black and white to show that the graphs are bipartite.
(a) An NOI-graph. 
(b) A COI-graph. 
(c) A GI graph. 
Solid and dashed edges indicate $Q$- and $L$-edges, respectively.}
\label{fig:composite-graph-examples}
\end{figure}


As before, we adapt the degree criteria and graph surgery procedures to hybrid graphs.  GI-graphs are too general for  Observation \ref{obs:applicability_of_degree_criterion_quick_check} to be applicable. Therefore, only NOI- and COI-graphs are considered. 
%
%
\begin{restatable}[Degree Criterion]{theorem}{degreecriterioncompositegraphs}
\label{thm:degree_criterion_composite_graphs_type_1}
If the density matrix $\rho(G)$ of a hybrid graph $G$ is separable and $G^\Gamma$ is a NOI- or a COI-graph, then
$ D(G) = D(G^\Gamma)$.
\end{restatable}
A proof for Theorem \ref{thm:degree_criterion_composite_graphs_type_1} is provided in Appendix \ref{apdx:composite_graphs}.

Graph surgery on a NOI-graph involves both $L$- and $Q$-surgeries. 
Any connected component in a NOI-graph has either all $L$-edges or all $Q$-edges. One can thus perform $L$- and $Q$-surgery independently on the respective connected components.  %

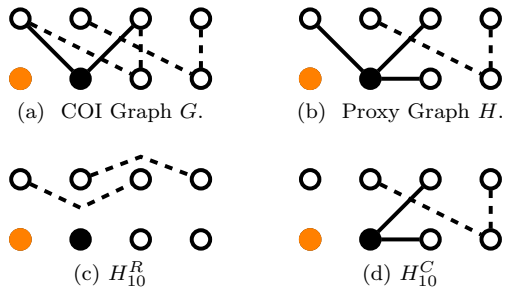
\begin{figure}
\subfloat[\label{subfig:coi_graph_surgery} COI Graph $G$.]{

    \begin{tikzpicture}[yscale=-1]
       \foreach \x in {1,...,{4}}
       {
        \foreach \y in {1,...,{2}}
        {
            \node[circle, draw=black, line width = 1.5pt, scale=0.8] (\y\x) at (0.8*\x, 0.8*\y) {};
        }
       }
       \node[circle, fill=black, scale=0.8] at (1.6,1.6) {};
       \foreach \v/\w in {11/22, 22/13}
   {
    \draw[-,line width = 1.5pt] (\v) -- (\w);
   }
   \foreach \v/\w in { 14/24, 12/24,11/23, 23/13}
   {
    \draw[dashed,line width = 1.5pt] (\v) -- (\w);
   }
 
   \node[circle, fill=orange, draw=orange, line width = 1.6pt, scale=0.8 ] at (0.8*1,0.8*2) {};
   
\end{tikzpicture}
}
\hfil
\subfloat[
    \label{subfig:coi_graph_surgery_noi} 
        Proxy Graph $H$.
    ]{
       
\begin{tikzpicture}[yscale=-1]
\foreach \x in {1,...,{4}}
{
    \foreach \y in {1,...,{2}}
    {
        \node[circle, draw=black, line width = 1.5pt, scale=0.8] (\y\x) at (0.8*\x, 0.8*\y) {};
    }
}
    \node[circle, fill=black, scale=0.8] at (0.8*2,0.8*2) {};
    
\foreach \v/\w in {11/22, 22/13, 22/23}
{
    \draw[-,line width = 1.5pt] (\v) -- (\w);
}

\foreach \v/\w in {14/24, 12/24}
{
    \draw[dashed,line width = 1.5pt] (\v) -- (\w);
}
\node[circle, fill=orange, draw=orange, line width = 1.6pt, scale=0.8 ] at (0.8*1,0.8*2) {};
\end{tikzpicture}
       
}
\hfil
\subfloat[$H^R_{10}$]{
       
\begin{tikzpicture}[yscale=-1]
\foreach \x in {1,...,{4}}
{
    \foreach \y in {1,...,{2}}
    {
        \node[circle, draw=black, line width = 1.5pt, scale=0.8] (\y\x) at (0.8*\x, 0.8*\y) {};
    }
}
    \node[circle, fill=black, scale=0.8] at (0.8*2,0.8*2) {};

\draw[dashed, line width=1.5pt] (11) -- (1.6, 1.18) -- (13);
\draw[dashed, line width=1.5pt] (12) -- (2.4, 0.50) -- (14);
\node[circle, fill=orange, draw=orange, line width = 1.6pt, scale=0.8 ] at (0.8*1,0.8*2) {};
\end{tikzpicture}
       
}
\hfil
\subfloat[$H^C_{10}$]{
       
\begin{tikzpicture}[yscale=-1]
\foreach \x in {1,...,{4}}
{
    \foreach \y in {1,...,{2}}
    {
        \node[circle, draw=black, line width = 1.5pt, scale=0.8] (\y\x) at (0.8*\x, 0.8*\y) {};
    }
}
    \node[circle, fill=black, scale=0.8] at (0.8*2,0.8*2) {};
    
\foreach \v/\w in {22/13, 22/23}
{
    \draw[-,line width = 1.5pt] (\v) -- (\w);
}

\foreach \v/\w in {14/24, 12/24}
{
    \draw[dashed,line width = 1.5pt] (\v) -- (\w);
}
\node[circle, fill=orange, draw=orange, line width = 1.6pt, scale=0.8 ] at (0.8*1,0.8*2) {};
\end{tikzpicture}
       
}
\caption{Vertices are colored in black and white to show that the graphs are bipartite.
Solid and dashed edges are $Q$- and $L$-edges, respectively. 
(a) Graph surgery on a COI-graph $G$ with vertex $(1, 0)$, colored orange, as the selected isolated vertex.
(b) Graph $H$, a proxy graph of $G$, as described in Section \ref{sec:hybrid_Laplacian}.
To derive $H$ from $G$, two $L$-edges $\{(0,0), (1,2)\}$ and $\{(0,2), (1,2)\}$ are removed and a $Q$-edge $\{(1,1), (1,2)\}$ is added.
(c) Graph $H^R_{10}$. 
Vertices $(0,0)$ and $(0,2)$ cannot be connected by a $Q$-edge because they belong to the same partition. So an $L$-edge is used.
(d) Graph $H^C_{10}$.
}
\label{fig:COI_graph_surgery}
\end{figure}

%
Graph surgery on a COI-graph however is not as straightforward. 
The non-identical STITCH steps of $L$- and $Q$-surgery are equally valid for any vertex with simultaneous incidences of $L$- and $Q$-edges.
This ambiguity is resolved by a proxy graph.

A \emph{proxy graph} of a COI-graph  is a NOI-graph such that the kernels of their hybrid Laplacians are identical.
It is constructed with a two-step process: first, by removing $L$-edges from all vertices on which both $L$- and $Q$- edges are incident; then, by reconnecting split connected components, if any, using $Q$-edges only.

\begin{restatable}{observation}{equivalentnonoverlappingincidencegraph}
\label{obs:equivalent_non_overlapping_incidence_graph}
Every COI-graph has a proxy graph.
\end{restatable}

The proof of Observation \ref{obs:equivalent_non_overlapping_incidence_graph} is found in Appendix \ref{apdx:composite_graphs}.
Deriving a proxy graph is akin to graph sparsification, which removes edges from a dense graph while preserving certain spectral properties of the Laplacian of the original graph \cite{graph_sparsification}.  
In the case of proxy graphs, only $L$-edges are removed and the preserved spectral property is the kernel of the hybrid Laplacian.
Given Observation \ref{obs:equivalent_non_overlapping_incidence_graph}, graph surgeries on a proxy NOI-graph and on the original COI-graph are equivalent.
Therefore, graph surgery on a COI-graph is performed  by first constructing a proxy NOI-graph and performing graph surgery on it. One iteration each of row and column surgeries on a COI-graph are shown in Fig. \ref{fig:COI_graph_surgery}.

The implication of graph surgery on hybrid graphs is the same as on $L$- and $Q$-graphs: If graph surgery on a hybrid graph always produces the empty graph, then the corresponding density matrix is entangled.

Hybrid graphs can also have weighted edges.
As in the case of weighted $L$- and $Q$-graphs, the degree criteria and the graph surgery procedures on unweighted hybrid graphs also apply to weighted hybrid graphs, as justified by the following lemma.

\begin{restatable}{lemma}{weightedcompositegraphsidenticalkernels}
\label{lem:weighted_composite_graphs_identical_kernels}
If the vertex and the edge sets of two weighted hybrid graphs are identical, their hybrid Laplacians have identical kernels.
\end{restatable}

The proof of Lemma \ref{lem:weighted_composite_graphs_identical_kernels} is found in Appendix \ref{apdx:composite_graphs}.

\section{\label{sec:bound_entangled} Construction of Bound Entangled States}
In
\cite{Lockhart2018, jh_phdthesis}, bound entangled $L$-grid sates are constructed using the degree criterion to verify a positive partial transpose of the density matrix, and the graph surgery procedure to verify entanglement.
This method can be used to construct new families of bound entangled states with the grid states presented in this paper.

\begin{restatable}{observation}{degreecriterionsufficient}
\label{obs:degree_criterion_sufficient}
If a grid-labelled graph $G$ satisfies $D(G) = D(G^\Gamma)$, the corresponding density matrix has a positive partial transpose, independent of whether the graph is interpreted as an $L$-graph, a $Q$-graph, a weighted graph, or a hybrid graph.
\end{restatable}

A proof of Observation \ref{obs:degree_criterion_sufficient} is provided in Appendix \ref{apdx:bound_entangled_states}.
According to the observation, the degree criterion verifies that a grid-state is positive under partial transpose, and graph surgery verifies that it is entangled.
Given that, bound entangled $Q$-grid states can be constructed using the degree criterion and the graph surgery procedure defined in Section \ref{sec:signless_Laplacians} if both the Q-graph and its partial transpose graph are bipartite. 
The cross-hatch pattern from  \cite{Lockhart2018} satisfies these conditions.
The pattern is in fact applicable not only to $Q$-graphs, but also to weighted and hybrid graphs.
\begin{restatable}{theorem}{crosshatchgraphsareentangled}
\label{thm:cross_hatch_graphs_are_entangled}
The density matrix of an $m \times n$  cross-hatch graph with $m, n \geq 3$ is  bound entangled for all grid states independent of whether the graph is interpreted as an $L$-graph, a $Q$-graph, a weighted graph, or a hybrid graph.
\end{restatable}

Theorem \ref{thm:cross_hatch_graphs_are_entangled} is proved in Appendix \ref{apdx:bound_entangled_states}.   Moreover, the cross-hatch pattern can be composed.
For example, irrespective of the Laplacian matrix used to interpret the resulting graph, a smaller cross-hatch graph can be embedded inside a bigger one as shown in Fig. \ref{fig:cross_hatch_example}(a) to produce new bound entangled states.
Likewise, the pattern can be tiled as shown in Fig. \ref{fig:cross_hatch_example}(b).
Both graphs in Fig. \ref{fig:cross_hatch_example} satisfy the degree criterion, because the constituent graphs in each graph individually satisfy the degree criterion. 
Therefore, they represent grid-states whose density matrices are positive under partial transpose.

Graph surgery on both graphs is carried out by first performing graph surgery on one of the  constituent graphs and  
then on the remaining edges of the other one. 
In the tiled composition, the STITCH step adds a diagonal edge, which can be treated as a part of another cross-hatch graph and be removed. 
In addition, the embedded and tiled compositions like in Fig. \ref{fig:cross_hatch_example} can also be composed to produce more bound entangled states, as long as the compositions satisfy the respective degree criterion and are reducible to empty graphs via graph surgery.
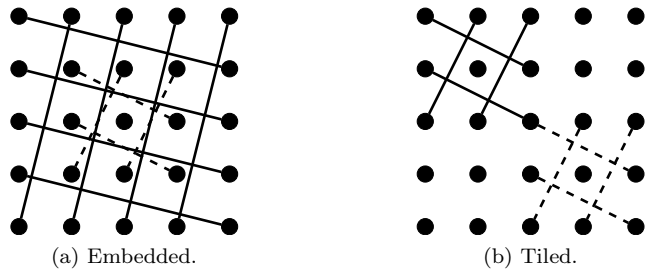
\begin{figure}
\subfloat[Embedded.]{

    \begin{tikzpicture}[yscale=-1]
        
   \foreach \x in {1,...,{5}}
   {
        \foreach \y in {1,...,{5}}
        {
            \node[circle, draw=black,fill, scale=0.7] (\y\x) at (0.7*\x, 0.7*\y) {};
            
        }
   }
   
   \foreach \v/\w in {22/34, 23/42, 24/43, 32/44}
   {
    \draw[dashed,line width = 1pt] (\v) -- (\w);
   }
   
   \foreach \v/\w in {11/25, 12/51, 13/52, 35/21, 15/54, 31/45, 14/53, 41/55}
   {
    \draw[-,line width = 1pt] (\v) -- (\w);
   }
\end{tikzpicture}
    \label{subfig:embedded_stacked_xhatch}
}
\hfill
\subfloat[Tiled.]{

    \begin{tikzpicture}[yscale=-1]
        
    \foreach \x in {1,...,{5}}
    {
        \foreach \y in {1,...,{5}}
        {
            \node[circle, draw=black, fill, scale=0.7] (\y\x) at (0.7*\x, 0.7*\y) {};
        
        }
    }
    
    \foreach \v/\w in {11/23, 12/31, 13/32, 33/21}
    {
        \draw[-,line width = 1pt] (\v) -- (\w);
    }
    
    \foreach \v/\w in {33/45, 43/55, 34/53, 35/54}
    {
        \draw[dashed, line width = 1pt] (\v) -- (\w);
    }

    
\end{tikzpicture}
    \label{subfig:tiled_staggered_xhatch1}
    
}

\caption{Examples of composing cross-hatch graphs. %
Solid and dashed edges distinguish constituent graphs.
The states corresponding to both graphs are bound entangled, irrespective of their interpretation as weighted or unweighted $L$- or $Q$-graphs, or as hybrid graphs.
}
\label{fig:cross_hatch_example}
\end{figure}

\section{Grid states corresponding to Hypergraphs}
With hybrid graphs, we showed that it is possible to generate density matrices from a mixture of $Q$- and $L$-edge states.
By defining a suitable Laplacian matrix, we derived degree criteria and graph surgery procedures.
As a proof of concept, we follow the same approach to extend the grid state model to hypergraphs.

Hypergraphs generalize graphs and allow edges to contain more than two vertices \cite{chung_hypergraph_matrix}.
Here, we only consider hypergraphs in which all hyperedges contain exactly three vertices.
In the literature, various approaches to extend graph matrices to hypergraphs are found, which range from matrices in \cite{hypergraph_matrix,  chung_hypergraph_matrix, adjacency_hypergraph} to tensors in \cite{hypergraph_tensor}. 
None of these previous approaches leads to a  density matrix that can be elegantly represented by a grid-labelled hypergraph.
Therefore, we first extend the notion of edge states and define hyperedge states, from which we define the density matrix and the hypergraph Laplacian matrix.
As such, the hyperedge state is chosen to be of the form  $1/\sqrt{3}\left(\ket{ij} + \ket{kl} + \ket{mn}\right)$.
The density matrix is the equal mixture of all hyperedge states in a hypergraph, and the Laplacian matrix is the unnormalized density matrix.
Split into a diagonal and a non-diagonal matrix, the Laplacian of a hypergraph $H$ is written as
\begin{equation}
    \label{eq:hypergraph-laplacian}
    L(H) = D(H) + A(H),
\end{equation}
where the diagonal matrix $D(H)$  and the non-diagonal matrix $A(H)$ matrix are the degree and adjacency matrices, respectively. 
The diagonal entries of the degree matrix are the degrees of vertices in the hypergraph.
The degree of a vertex is the number of hyperedges incident on the vertex.
The adjacency matrix is defined as
\begin{equation}
    A_{\alpha \beta} = 
    \begin{cases} 
        \text{adj}(v_\alpha, v_\beta), \text{ if } \alpha \neq \beta, \\
        0, \text{ otherwise,}
    \end{cases}
\end{equation}
where $\text{adj}(v_\alpha, v_\beta)$ is the number of hyperedges connecting vertices $v_\alpha$ and $v_\beta$.

\subsection{Weighted Graph Model for hypergraph}
A hypergraph can be modeled with a weighted graph, and its Laplacian matrix can be connected to the signless Laplacian matrix of the weighted graph.

Consider a hypergraph $H$ with two hyperedges in Figs. \ref{fig:hypergraph_weighted_model}(a, b). 
Each hyperedge is turned into a clique as in Figs. \ref{fig:hypergraph_weighted_model}(c, d).
A clique is a subset of vertices of a graph such that every vertex in the set is connected to every other vertex in the set \cite{bondy2009graph}.
The cliques are combined into a weighted graph as in Fig. \ref{fig:hypergraph_weighted_model}(e) such that the edge weight of an edge connecting a vertex pair is the cumulative number of edges in all cliques that connect the vertex pair.
In Fig. \ref{fig:hypergraph_weighted_model}(e), the weights of black edges are all 1 and the orange edge is weighted 2.
We call the weighted graph derived in this fashion the graph of a hypergraph.
Formally, the graph of a hypergraph $H$ is a weighted graph $G$ such that any vertex pair $\{v_\alpha, v_\beta\}$ connected by a hyperedge in $H$ is connected in $G$ by an edge with weight $A(H)_{\alpha \beta}$.

With this construction, the adjacency matrix of a hypergraph and of its graph are the same matrix.
But the degree matrices are different.
Consider a hypergraph $H$ and its graph $G$. 
The degree of a non-isolated vertex $v_\alpha$ in $H$ is $D(H)_\alpha < \sum_\beta A(H)_{\alpha \beta}$.
However, in the graph $G$ the degree of the same vertex by definition is $D(G)_\alpha = \sum_\beta A(G)_{\alpha \beta}$.
The degree matrices of a hypergraph and its graph thus are offset by a diagonal non-negative matrix, which we call the offset matrix and define it as
\begin{equation}
    O(H) = D(G) - D(H),
\end{equation}
where $H$ is a hypergraph, $G$ its graph, and $O(H)$ the offset matrix.
With these observations, the hypergraph Laplacian of a hypergraph $H$ can be written as
\begin{equation}
\label{eq:hypergraph-laplacian-with-offset-matrix}
    L(H) = Q(G) - O(H),
\end{equation}
where $Q(\cdot)$ indicates the signless Laplacian.

With the weighted graph model, we can derive a degree criterion for hypergraph grid states.

\begin{restatable}{theorem}{hypergraphdegreecriterion}
\label{thm:hypergraph_degree_criteion}
 Let $H$ be a hypergraph and $G$ be its graph. If $\rho(H)$ is separable and $G^\Gamma$ is bipartite, then
 \(
 D(G) = D(G^\Gamma).
 \)
\end{restatable}

For the proof of Theorem \ref{thm:hypergraph_degree_criteion}, see Appendix \ref{apdx:hypergraphs}.
Unlike the degree criteria for grid-labelled graphs, it is not clear that the hypergraph degree criterion is sufficient for the positive partial transpose of the hypergraph density matrix. 
Suppose $H$ is a hypergraph and $G$ is its graph, and $D(G) = D(G^\Gamma)$.
Then,
\begin{align}
     Q^\Gamma(G) &= D^\Gamma(G^\Gamma) + A^\Gamma(G) \nonumber \\
    &= D(G^\Gamma) + A(G^\Gamma)  = Q(G^\Gamma),
\end{align}
and from Equation (\ref{eq:hypergraph-laplacian-with-offset-matrix})
\begin{align}
     Q(G) &= L(H) + O(H). \nonumber \\
     \implies Q^\Gamma(G) &= L^\Gamma(H) + O^\Gamma(H) 
     = L^\Gamma(H) + O(H),
\end{align}
from which it follows
\begin{equation}
\label{eq:hypergrah-laplacian-partial-transpose}
L^\Gamma(H) = Q(G^\Gamma) - O(H).
\end{equation}
From Equation (\ref{eq:hypergrah-laplacian-partial-transpose}), it is not
clear if $Q(G^\Gamma) - O(H)$ is always positive semidefinite.
On the other hand, consider a $2 \times 2$ hypergraph $H$ with a single hyperedge shown in Fig. \ref{fig:hypergraph_weighted_model}(a).
The graph $G$ of the hypergraph is the graph in Fig. \ref{fig:hypergraph_weighted_model}(c).
It is easily seen that $D(G) \neq D(G^\Gamma)$, and also verified using the PPT criterion that $\rho(H)$ is entangled.

Graph surgery cannot be extended to hypergraphs via the weighted graph model. 
The graph surgery procedure for weighted $Q$-graphs requires the graphs to be bipartite. 
The graph of a hypergraph, although a weighted $Q$-graph, is not bipartite, because cliques are inherently not bipartite.

\begin{figure}
\centering
\subfloat[\label{subfig:hypergraph_example_hyperedge_1}]{
    \begin{tikzpicture}[yscale=-1]
    \tikzset{
    contour/.style={
      black, 
      line width=1pt,
      double distance=5mm,
      line cap=round,
      rounded corners,
        }
    }
       \foreach \x in {0,...,{1}}
       {
            \foreach \y in {0,1}
            {

                \node[circle, draw=black, fill=black,line width = 1pt, scale=0.7 ] (\y\x) at (0.7*\x, 0.7*\y) {};
            }
       }
       \begin{pgfonlayer}{bg}
        \draw[contour](00.center)--(10.center)--(11.center);
  \end{pgfonlayer}
   \end{tikzpicture}
}
\hfill
\subfloat[]{
    \begin{tikzpicture}[yscale=-1]
    \tikzset{
    contour/.style={
      black, 
      line width=1pt,
      double distance=5mm,
      line cap=round,
      rounded corners,
        }
    }
       \foreach \x in {0,...,{1}}
       {
            \foreach \y in {0,1}
            {

                \node[circle, draw=black, fill=black,line width = 1pt, scale=0.7 ] (\y\x) at (0.7*\x, 0.7*\y) {};
            }
       }
       \begin{pgfonlayer}{bg}
        \draw[contour](01.center)--(11.center)--(10.center);
  \end{pgfonlayer}
   \end{tikzpicture}
}
\hfill
\subfloat[]{
    \begin{tikzpicture}[yscale=-1]
       \foreach \x in {0,...,{1}}
       {
            \foreach \y in {0,1}
            {

                \node[circle, draw=black, fill=black,line width = 1pt, scale=0.7 ] (\y\x) at (0.7*\x, 0.7*\y) {};
            }
       }
       \begin{pgfonlayer}{bg}
        \foreach \v/\w in { 11/00,11/10,00/10}
       {
        \draw[-,line width = 1pt] (\v) -- (\w);
       }
  \end{pgfonlayer}
   \end{tikzpicture}
}
\hfill
\subfloat[]{
    \begin{tikzpicture}[yscale=-1]
       \foreach \x in {0,...,{1}}
       {
            \foreach \y in {0,1}
            {

                \node[circle, draw=black, fill=black,line width = 1pt, scale=0.7 ] (\y\x) at (0.7*\x, 0.7*\y) {};
            }
       }
       \begin{pgfonlayer}{bg}
        \foreach \v/\w in { 01/10,11/10,01/11}
       {
        \draw[-,line width = 1pt] (\v) -- (\w);
       }
  \end{pgfonlayer}
   \end{tikzpicture}
}
\hfill
\subfloat[]{
    \begin{tikzpicture}[yscale=-1]
   \foreach \x in {0,...,{1}}
   {
    \foreach \y in {0,1}
    {

        \node[circle, fill=black, draw=black,line width = 1pt, scale=0.7 ] (\y\x) at (0.7*\x, 0.7*\y) {};
    }
   }
    \draw[black, line width =1pt](00)--(10);
    \draw[black, line width =1pt](00)--(11);
    \draw[orange, line width =1.5pt](11)--(10);
    \draw[black, line width =1pt](11)--(01);
    \draw[black, line width =1pt](01)--(10);
\end{tikzpicture}
}
\caption{Weighted graph model of a hypergraph. (a) and (b) Two hyperedges. (c) and (d) Their respective cliques. (e) Weighted graph derived from the cliques.}
\label{fig:hypergraph_weighted_model}
\end{figure}
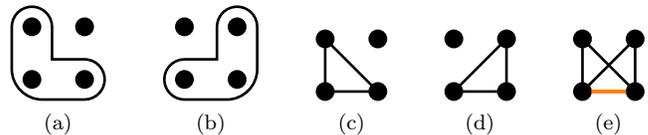

Even though this interpretation of hypergraph grid states does not allow graph surgery, it illustrates the flexibility of the grid-state model. 
We were not only able to define a hypergraph laplacian matrix in an ad-hoc manner to suit our purpose, but also integrate the weighted Laplacian to derive a degree criterion for hypergraph grid states.

\section{Conclusion and Outlook}
This paper reveals a rich interplay between graphs and quantum states.
Using a variety of interpretations of graphs as density matrices, we have identified additional families of grid states beyond the ones originally suggested in \cite{Lockhart2018} and shown that their entanglement properties relate to properties of the corresponding graphs. 
In particular, we investigated signless Laplacians and weighted graphs.
We introduced the concept of hybrid graphs, containing two different types of edges, and derived the entanglement properties of the corresponding grid states.
Additionally, we constructed new families of bound entangled states with these new grid states, using the method from \cite{Lockhart2018}. 
We showed that the cross-hatch pattern is not only bound entangled for the new families of grid states, but it could also be composed to construct more bound entangled states.
We noted two additional links between graph theory and grid states -- resemblance between hybrid graphs and signed graphs, and between proxy graph construction and graph sparsification. 
Further work into these links would be interesting. 
For example, one could investigate if proxy graphs can be connected to the concept of local graph isomorphism discussed in \cite{jh_phdthesis}.

We demonstrated with hypergraph grid states that our approach for hybrid graphs can be applied in other contexts.
Similar approaches could be used to incorporate more general edge states, for example, with the normalized Laplacian defined in  \cite{chung1997spectral} and with complex Laplacian matrices. 

Since genuine multipartite entanglement has been found in $L$-grid states \cite{Lockhart2018},
for further work, one could investigate if the same is the case for grid states presented above.
Finally, as the graph surgery procedure is not possible without isolated vertices, 
it would be desirable to improve graph surgery or find alternative procedures that do not require isolated vertices.

\begin{acknowledgments}
We acknowledge financial support by the QuantERA grant QuICHE and the German ministry of education and research (BMBF, grant no. 16KIS1119K), and by Deutsche Forschungsgemeinschaft (DFG, German Research Foundation) under Germany's Excellence Strategy – Cluster of Excellence Matter and Light for Quantum Computing (ML4Q) EXC 2004/1 – 390534769.  
\end{acknowledgments}

\appendix
\counterwithin{lemma}{section}
\counterwithin{observation}{section}
\counterwithin{corollary}{section}

\section{Additional Graph Theory Concepts}
\label{apdx:additional_graph_theory_concepts}

The appendices contain the proofs of results stated in the main text.
The  statements are repeated before each proof.
%
In this section, in addition to the proof of Observation \ref{obs:applicability_of_degree_criterion_quick_check}, we present graph concepts used in the proofs.

The \emph{unoriented incidence matrix} of a graph $G = (V, E)$ is the $\abs{V} \times \abs{E}$ matrix $R$
such that $R_{ij} = \sqrt{w_j}$ if edge $e_j$ with weight $w_j$ is incident on vertex $v_i$, and $R_{ij} = 0$ otherwise.
The \emph{oriented incidence matrix} $B$ results from negating one of the two non-zero entries in each column of the matrix $R$ \cite{Mohar1997}.
The signed and the signless Laplacian matrices satisfy  $L = BB^T$ and $Q = RR^T$ \cite{diestel2010, signlessLapla}. 

For the proof of Observation \ref{obs:applicability_of_degree_criterion_quick_check}, the following lemma is needed. 
Hereafter, $K(\cdot)$ denotes the kernel of a matrix.
\begin{lemma}[\cite{Braunstein2006}]
\label{lem:degree_criterion_support_lemma}
Let $M$ and $\Delta$ be $n \times n$ real matrices. 
Let $M$ be symmetric and positive semidefinite,   and  $\Delta$ be  nonzero, diagonal, and traceless. 
If a vector $\mathbf{v} \in \{-1, 1\}^n$ exists in $K(M)$, then $M + \Delta \ngeq 0$.
\end{lemma}

\begin{proof}[Proof of Lemma \ref{lem:degree_criterion_support_lemma} ]

  Given the nature of matrix $\Delta$, at least one of its diagonal entries, say $\Delta_{ii} = \delta$, is positive and nonzero.
  Let $\mathbf{v} \in \{-1, 1\}^n$ be in $K(M)$. Let
\(
\mathbf{w} \coloneqq \mathbf{v} + a \mathbf{x},
\)
with $a \in \mathbb{R}$, and $\mathbf{x}$ be the $i$-th standard basis vector.
Consider the inner product
\begin{align}
I \coloneqq& \langle \mathbf{w}, (M + \Delta)\mathbf{w} \rangle \nonumber \\
    =& \langle \mathbf{v}, M\mathbf{v}\rangle 
    + a \langle \mathbf{v}, M\mathbf{x}\rangle 
    + a \langle \mathbf{x}, M\mathbf{v}\rangle + a^2 \langle  \mathbf{x},  M\mathbf{x}  \rangle \notag \\
    ~+&  \langle   \mathbf{v}, \Delta\mathbf{v}  \rangle + a  \langle   \mathbf{v}, \Delta\mathbf{x}  \rangle  + a \langle   \mathbf{x}, \Delta\mathbf{v}  \rangle  + a^2 \langle   \mathbf{x}, \Delta\mathbf{x}  \rangle. \label{eq:degre_criterion_abstraction_main_inner_product}
\end{align}
The scalars $\langle \, \mathbf{v}, M\mathbf{v} \,\rangle$,  $ \langle \, \mathbf{v}, M\mathbf{x}\, \rangle$ and $\langle \, \mathbf{x}, M\mathbf{v}\, \rangle$ are all 0, because  $M\mathbf{v}=0$. 
 And,  $\langle \, \mathbf{x}, M\mathbf{x}\, \rangle = M_{ii}$ and 
 $ \langle \, \mathbf{x}, \Delta\mathbf{x}\rangle = \delta$. 
 The remaining terms are 
 \begin{align}
 \langle \, \mathbf{v}, \Delta\mathbf{v}\, \rangle &= \sum_{j=1}^n\left(\mathbf{v}_j\right)^2 \Delta_{jj} = \tr(\Delta) = 0  \\
 \shortintertext{and}
 \langle \, \mathbf{v}, \Delta\mathbf{x}\, \rangle  &= \langle \, \mathbf{x}, \Delta\mathbf{v}\, \rangle =  \pm\delta,\,  \text{ if } \mathbf{v}_i = \pm1.
\end{align}
\noindent
Equation (\ref{eq:degre_criterion_abstraction_main_inner_product}) thus reduces to
\begin{equation}
\label{eq:degree-criterion-helper-lemma-final-condition}
    I=  a^2(\delta + M_{ii}) \pm 2\delta a,\,  \text{ if } \mathbf{v}_i = \pm1.
\end{equation}
Notice that all diagonal entries of the matrix $M$ are non-negative, because $M$ is positive semidefinite.
Equation (\ref{eq:degree-criterion-helper-lemma-final-condition}) therefore always has distinct roots, because $M_{ii} + \delta > 0$.
This implies that there exists $a$ for which
$I < 0$, meaning
\(
\label{eq:Q_degree_criterion_from_PPT}
M + \Delta \ngeq 0.
\)
\end{proof}

We now prove the observation.

\degreecriterionquickcheck*


\begin{proof}[Proof of Observation \ref{obs:applicability_of_degree_criterion_quick_check}]
    Let $G$ be a grid-labelled graph on $n$ vertices, and $D(G)$ and $A(G)$ be the degree and the adjacency matrices of $G$, respectively. 
    Let $\mathrm{L}(G) = D(G) \pm A(G)$ be a generic Laplacian matrix representative of the Laplacian matrices used in this paper. 
    Let the corresponding density matrix $\rho(G)$ be the normalized $\mathrm{L}(G)$.
    Then
    \[
    \mathrm{L}^{\Gamma}(G) = D^{\Gamma}(G) \pm A^{\Gamma}(G) = D(G) \pm A(G^{\Gamma}),
    \]
    which implies
    \begin{align}
    \label{eq:Q_degree_criterion_transposition}
    \mathrm{L}^{\Gamma}(G) &= D(G) + \mathrm{L}(G^{\Gamma}) - D(G^{\Gamma}) \\ \nonumber
        &= \mathrm{L}(G^{\Gamma}) + \Delta, \nonumber 
    \end{align}
    where the matrix $\Delta = D(G) - D(G^{\Gamma})$ is traceless and diagonal.
    If $\Delta$ is nonzero, then since $\mathrm{L}(G^\Gamma) \geq 0$, Lemma \ref{lem:degree_criterion_support_lemma} implies
    $
    \mathrm{L}(G^{\Gamma}) + \Delta \ngeq 0.
    $
    But $\rho(G)$ is separable and the PPT criterion requires $\rho^\Gamma(G) \geq 0$, meaning $\mathrm{L}^{\Gamma}(G) \geq 0$. This is a contradiction. Then, it must be that $\Delta = D(G) - D(G^{\Gamma}) = 0$.
    \end{proof}

\section{$Q$-Grid States}
\label{apdx:signless_Laplacian}

Proof of results stated in Section \ref{sec:signless_Laplacians} are given here.
Several supporting observations are needed for the proof of Lemma \ref{lem:singless_Laplacian_null_space_only_if_bipartite}, which is then used to prove the degree criterion.
\begin{observation}[\cite{signlessLapla}]
\label{obs:connected_q_graph_bipartite}
The least eigenvalue of the signless Laplacian of a connected graph is equal to 0 if and only if the graph is bipartite. In this case 0 is a simple eigenvalue.
\end{observation}
Next, we deduce a property of the kernel of the signless Laplacian matrix of connected bipartite graphs.
\begin{observation}
\label{obs:signless_Laplacian_null_space_vector_nature}
For any connected bipartite graph $G$ on $n$ vertices there exists a vector $\mathbf{v} \in \{-1,1\}^{n}$ in $K[Q(G)]$.
\begin{proof}[Proof of Observation \ref{obs:signless_Laplacian_null_space_vector_nature}]   
Let the two vertex partitions in $G$ be $P_1$ and $P_2$.
From Observation \ref{obs:connected_q_graph_bipartite}, $Q(G)$ has a non-trivial kernel because $G$ is bipartite.
Suppose a vector $\mathbf{v} \in \{-1, 1\}^n$ is constructed as follows: if the $k^{\text{th}}$ vertex is in ${P_1}$ then the component $v_k=1$, otherwise $v_k=-1$.
Given that the vertices connected by any edge in $G$ belong to opposite partitions, from the definition of the incidence matrix $R$, we see that $R(G)^T \mathbf{v} = \mathbf{0}$.
Then $Q(G) \mathbf{v} = R(G) R(G)^T \mathbf{v} = \mathbf{0}$.
\end{proof}
\end{observation}
Finally, with another result from \cite{signlessLapla}, we derive a corollary to prove Lemma \ref{lem:singless_Laplacian_null_space_only_if_bipartite}.
\begin{observation}[\cite{signlessLapla}]
\label{obs:q_null_space_multiplicity}
In any graph, the (algebraic) multiplicity of the eigenvalue 0 of the signless Laplacian is equal
to the number of bipartite (connected) components.
\end{observation}
\begin{corollary}
\label{coro:bipartite-graph-kernel-generalization}
Each connected component in a bipartite graph $G$ on $n$ vertices corresponds to a basis vector $\mathbf{v} \in \{-1, 0,1\}^n$ of $K[Q(G)]$.
\end{corollary}
\begin{proof}[Proof of Corollary \ref{coro:bipartite-graph-kernel-generalization}]
Observation \ref{coro:bipartite-graph-kernel-generalization} applies to connected components, because they are connected subgraphs. 
If $G$ is not a connected graph, the vectors from Observation \ref{coro:bipartite-graph-kernel-generalization} are extended by
setting vector components to 0 for vertices not in the connected component.
Let $\mathbf{v}_k$ denote the vector associated in this way to the connected component $C_k$ of $G$. 
Then the set of  vectors $\{\mathbf{v}_k\}$  is linearly independent, because the vectors have disjoint support.

Since $Q(G)$ is diagonalizable, the algebraic and geometric multiplicities of its eigenvalues are equal \cite{matrix_analysis}.
From Observation \ref{obs:q_null_space_multiplicity} and the previous statement, the geometric multiplicity of the 0 eigenvalue of $Q(G)$ is the number of connected components in $G$, which is equal to the cardinality of $\{\mathbf{v}_k  \}$.
Suppose $\abs{\{\mathbf{v}_k  \}} = m$. 
Then we have $m$ linearly independent vectors in the $m$-dimensional kernel of $Q(G)$.
The vectors therefore span $K[Q(G)]$.
\end{proof}

The next lemma allows us to use Observation \ref{obs:applicability_of_degree_criterion_quick_check} on $Q$-graphs.
\begin{lemma}
    \label{lem:singless_Laplacian_null_space_only_if_bipartite}
    A vector $\mathbf{v} \in \{-1,1\}^n$ exists in the kernel $Q(G)$ of a graph $G$ on $n$ vertices if and only if it is bipartite.
\end{lemma}

\begin{proof}[Proof of Lemma \ref{lem:singless_Laplacian_null_space_only_if_bipartite}]
Let $G$ be a bipartite graph on $n$ vertices. 
Let $\{\mathbf{v}_k\}$ be vectors derived from connected components, including isolated vertices, of $G$ as described in Corollary \ref{coro:bipartite-graph-kernel-generalization}.
Then, because the vectors $\{\mathbf{v}_k\}$ have disjoint support,  the sum
\(
 \sum_k \mathbf{v}_k \eqqcolon \mathbf{v} \in \{-1,1\}^n \text{ and } Q(G) \mathbf{v} = 0.
\)

If a vector $\mathbf{v} \in \{-1, 1\}^n$ is in  $K[Q(G)]$, then $Q(G)\mathbf{v}=0$,  meaning $R^T \mathbf{v} =0$.
It then follows from Proposition 2.1 in \cite{signlessLapla} that $G$ is bipartite.
\end{proof}

Finally, we prove the degree criterion for $Q$-graphs.

\degreecriterionqgraphs*


\begin{proof}[Proof of Theorem \ref{thm:degree_criterion_Q}]
    Using Lemma \ref{lem:singless_Laplacian_null_space_only_if_bipartite},  the proof follows from applying Observation \ref{obs:applicability_of_degree_criterion_quick_check} to $Q$-graphs. 
\end{proof}

For the proof of Observation \ref{obs:graph_surgery_q_graphs}, we assign a notion of vectors to vertices in a grid-labelled graph.
The vector of a vertex is the standard basis vector corresponding to its index.
In a grid-labelled graph, the vertices are indexed row-wise from top-left to bottom-right. 
Thus, in an $m \times n$ grid-labelled graph, the vertex $(0,0)$ is the first vertex and is assigned the standard basis vector $\mathbf{e}_1$.
The vertex $(m-1, n-1)$ is the last vertex and is assigned the vector $\mathbf{e}_{m \cdot n}$.
This is convenient because the state vector of the state $\ket{0 \, 0}$ is $\mathbf{e}_1$ and of $\ket{m-1 , n-1}$ is $\mathbf{e}_{mn}$. With this convention, we can say vertex $(i,j)$ corresponds to the state $\ket{i \, j}$.

\graphsurgeryqgraphs*

    

\begin{proof}[Proof of Observation \ref{obs:graph_surgery_q_graphs}]
Given vertex $(i, j)$ is an isolated vertex and thus a connected component, by
Corollary \ref{coro:bipartite-graph-kernel-generalization},
$\rho_Q(G)\ket{i  \, j} = 0$. 
Since $\rho_Q(G)$ is hermitian,   $\ip{\mu \, \nu}{i \, j} = 0$, which 
implies either $\ip{i}{\mu} = 0$ or $\ip{j}{\nu} = 0$.
We first consider the case $\ip{i}{\mu} = 0$, from which it follows that 
the inner product $\ip{\mu \, \nu}{i\, j_c} = 0$ for all $c$.
This means $\ket{\mu \, \nu}$ is orthogonal to states corresponding to all vertices in row $i$.

Let $C_k$ be a connected component in $G$ and $\ket{C_k} \coloneqq \mathbf{v}_k$ be the  basis vector from Corollary \ref{coro:bipartite-graph-kernel-generalization} of $K[\rho_Q(G)]$.
%
%
%
Then $\ip{\mu \, \nu}{C_k} = 0$.

Consider the vector $\ket{C'_k} \coloneqq \ket{C_k} + \ket{L}$, where  $\ket{L} \coloneqq \sum_c \lambda_c \ket{i \, j_c}$ is a linear combination of vectors of all vertices in row $i$.
%
A suitable set of scalars $\{\lambda_c\}$ can always be chosen to make $\ket{C'_k}_c = 0$ for all $c$.  
Using Corollary \ref{coro:bipartite-graph-kernel-generalization}, the vector $\ket{C'_k}$ can be interpreted as the vector of a connected component $C'_k$ that includes all vertices in $C_k$ except the ones in row $i$.
Vertices in $C'_k$ have the same relative partitioning as in $C_k$.
Further,  $\ip{\mu \, \nu}{C'_k} = 0$, because $\ip{\mu \, \nu}{L} = 0$ as $\ip{\mu \, \nu}{i\, j_c} = 0$ for all $c$, and $\ip{\mu \, \nu}{C_k} = 0$.

Let $G'$ be a grid-labelled graph with the same vertex set as $G$.
For every connected component $C_k$ in $G$, let the graph $G'$ have the connected component $C'_k$ derived from $C_k$ as described above.
Notice that the isolated vertices $\{(i_o, j_o)\}$ in $G$ remain isolated in $G'$, and that $G'$ has additional isolated vertices -- the vertices in row $i$.
The graph $G'$ thus can be produced via row surgery on $G$ with isolated vertex $(i, j)$. 
It can therefore be labelled as $G^R_{ij}$.

Depending the nature of the vectors $\{\ket{C'_k}\}$, we have three possibilities:
\begin{itemize}
    \item If the vectors $\{\ket{C'_k}\}$ are all in $\{1, 0\}^n$, then $G^R_{ij}$ is an $L$-graph.
    The kernel of $L(G^R_{ij})$ is spanned by the vectors  $\{\ket{C'_k}\}$, $\{\ket{i_o\, j_o}\}$, and $\{\ket{i, j_c}\}$ of its connected components, to all of which $\ket{\mu \, \nu}$ is orthogonal.
    Thus $\ket{\mu \, \nu}$ is in the range of $L(G^R_{ij})$ and also of $\rho_L(G^R_{ij})$.
    This case is identical to $L$-surgery.
    \item If the vectors $\{\ket{C'_k}\}$ are all in $\{1, 0, -1\}^n$, then by Corollary \ref{coro:bipartite-graph-kernel-generalization} and arguments analogous to above, the vector $\ket{\mu \, \nu}$ is in the range of $\rho_Q(G^R_{ij})$.
    \item Finally, if some vectors in $\{\ket{C'_k}\}$ are in $\{1, 0\}^n$ and others in $\{1, 0, -1\}^n$, then $G'$ is a hybrid graph.
    Graph surgery on hybrid graphs are presented in Section \ref{sec:hybrid_Laplacian}.
\end{itemize}

 It can be shown with analogous arguments that if instead $\ip{l}{j} = 0$, then $\ket{\mu \, \nu}$ is in the range of $\rho_L(G^C_{ij})$ or of $\rho_Q(G^C_{ij})$ or  of the density matrix of an analogous hybrid graph.
\end{proof}
We now show the unitary inequivalence of the $L$- and the $Q$-grid states corresponding to the same non-bipartite grid-labelled graph.

\qlgraphsunitarytransformation*


\begin{proof}[Proof of Observation \ref{obs:Q_L_graphs_unitary_transformation}]
    Let $G$ be a non-bipartite grid-labelled graph.
    The dimension of $K[L(G)]$ is the number of connected components in $G$ (see Section 3.13.5 in \cite{brouwer2011spectra}).
    From Corollary \ref{coro:bipartite-graph-kernel-generalization}, the dimension of $K[Q(G)]$ is the number of bipartite connected components in $G$.
    At least one connected component in $G$ is not bipartite.
   This means the dimensions of $K[L(G)]$ and of $K[Q(G)]$ are not equal.
    Then from the rank-nullity theorem, the ranks of $L(G)$ and of $Q(G)$ are not equal.
    Therefore, $\rho_L(G)$ and $\rho_Q(G)$ cannot be unitarily equivalent.
\end{proof}

\section{Weighted Graphs}
\label{apdx:weighted_graphs}

This section consists of proof of results stated for weighted grid-labelled graphs in the main text.

\nullspaceofweightedunweightedgraphlaplacianconincice*


\begin{proof}[Proof of Lemma \ref{lem:null_space_of_weighted_unweighted_graph_Laplacian_conincice}]
Let $G = (V, E)$ be a weighted graph and  edge weights of edges in $G$ be $\{w_1, \dots w_m \}$, where $m=\abs{E}$.
If $Q\mathbf{v} = \mathbf{0}$, then 
\begin{equation}
\label{eq:weighted_incidence_matrix_null_space}
\big[R^T\mathbf{v}\big]_i = 
    \sqrt{w_i}\left(\mathbf{v}_{i1} + \mathbf{v}_{i2}\right)
     = \mathbf{0}, \, \forall i \in \{1,\dots, m\},
\end{equation}
because $Q = RR^T$, where $R$ is the unoriented incidence matrix.
The vector components $\{\mathbf{v}_{i1}, \mathbf{v}_{i2}\}$ correspond to vertices connected by edge $e_i \in E$.
The solutions of Equation (\ref{eq:weighted_incidence_matrix_null_space}) are independent of the edge weights.
Therefore, any vector $\mathbf{v} \in K[Q(G)]$ must also be in the kernels $\{K[Q(G')]\}$ of all graphs $\{G'\}$ with the same edge and vertex sets. 
The same arguments apply to the signed Laplacian $L(G)$.
\end{proof}

\weightedqlgraphsunitarytransformation*


\begin{proof}[Proof of Corollary \ref{coro:weighted_Q_L_graphs_unitary_transformation}]
Let $G_w = (V, E)$ be a non-bipartite weighted grid-labelled graph and $G = (V, E)$ be its unweighted counterpart.
From the proof Observation \ref{obs:Q_L_graphs_unitary_transformation}, we know $\rho_L(G)$ and $\rho_Q(G)$ are not unitarily equivalent because their ranks are not equal.
According to Lemma \ref{lem:null_space_of_weighted_unweighted_graph_Laplacian_conincice}, $K[\rho_L(G)] = K[\rho_L(G_w)]$ and $K[\rho_Q(G)] = K[\rho_Q(G_w)]$.
This means that the ranks of $\rho_L(G_w)$ and $\rho_Q(G_w)$ are not equal. 
Therefore, the density matrices cannot be unitarily equivalent.

\end{proof}

\section{Hybrid Graphs}
\label{apdx:composite_graphs}

Here, we prove the results for grid-states derived from the grid-labelled hybrid graphs.
To proceed, we need a notion of incidence matrix.
The \emph{incidence matrix} of a hybrid graph $G = (V, E)$ is the $\abs{V} \times
\abs{E}$ matrix 
\(
\mathcal{R} = \bm{[}B_l \quad R_q\bm{]},
\)
where $B_l$ and $R_q$ are the unoriented and the oriented incident matrices of its $L$- and $Q$-subgraphs, respectively.
The hybrid Laplacian satisfies $\mathcal{L} = \mathcal{R}\mathcal{R}^T$.

Like in the case of $Q$-grid states, we need supporting lemmas to prove the degree criterion for NOI- and COI-graphs.
\begin{lemma}
    \label{lem:NOI_COI_kernel_connected_components}
    Each connected component in a NOI- or a COI-graph $G$ on $n$ vertices corresponds to a basis vector $\mathbf{v} \in \{-1, 0,1\}^n$ of $K[\mathcal{L}(G)]$.
\end{lemma}
\begin{proof}[Proof of Lemma \ref{lem:NOI_COI_kernel_connected_components}]
    The proof follows for adapting the arguments in the proof of Corollary \ref{coro:bipartite-graph-kernel-generalization} to NOI- and COI-graphs.
\end{proof}

\begin{lemma}
    \label{lem:null_space_vector_COI_NOI}
    For any NOI- or COI-graph $G$ on $n$ vertices there exists a vector $\mathbf{v} \in \{-1,1\}^{n}$ in the kernel of $\mathcal{L}(G)$.
\end{lemma}
\begin{proof}[Proof of Lemma \ref{lem:null_space_vector_COI_NOI}]
With Lemma \ref{lem:NOI_COI_kernel_connected_components}, arguments analogous to the ones given in the proof Lemma \ref{lem:singless_Laplacian_null_space_only_if_bipartite} prove this lemma.
%

%
\end{proof}

\degreecriterioncompositegraphs*


\begin{proof}[Proof of Theorem \ref{thm:degree_criterion_composite_graphs_type_1}]
Using Lemma \ref{lem:null_space_vector_COI_NOI}, the proof follows from applying Observation \ref{obs:applicability_of_degree_criterion_quick_check} to a NOI- or a COI-graph.
\end{proof}

We now prove the claim that every COI graph has a proxy graph.

 \equivalentnonoverlappingincidencegraph*
 

\begin{proof}[Proof of Observation \ref{obs:equivalent_non_overlapping_incidence_graph}]
Let $G$ be a COI-graph with two vertex partitions $P_1$ and $P_2$ determined by its $Q$-subgraph.

First, note that any connected component that contains a $Q$-edge must contain at least one vertex in partition $P_1$, since $Q$-edges connect vertices in opposite partitions. 
Second, by definition, the pair of vertices connected by any $L$-edge in $G$ must both be in the same partition. 
Using these observations, we can construct the proxy graph as follows:

\begin{itemize}
\item For each connected component that contains a $Q$-edge, choose two designated vertices -- one in partition $P_1$ and the other is partition $P_2$. 
\item Then, for all vertices in the graph that have both an $L$-edge and a $Q$-edge incident, remove the $L$-edge.
\item If a vertex belonging to partition $P_1$ (resp. $P_2$) is isolated from its previous connected component, reconnect it with a $Q$-edge to the corresponding designated vertex in partition $P_2$ (resp. $P_1$). 
\end{itemize}

The above steps not only yield a NOI-graph, say $G'$, but also guarantee that the relative vertex partitioning of the vertices in $G$ and in $G'$ remain identical, and that all connected components in $G'$ have the same vertices as in their counterpart in $G$. 
Therefore, the vectors associated to connected components in $G$ and to connected components in $G'$ are identical.
Then, from Lemma \ref{lem:NOI_COI_kernel_connected_components}, it follows that the kernels of $\mathcal{L}(G)$ and of $\mathcal{L}(G')$ are identical.
\end{proof}

Finally, we show that in the case of hybrid graphs as well the edge weights alone do not affect the kernel of the hybrid laplacian.

\weightedcompositegraphsidenticalkernels*

\begin{proof}[Proof of Lemma \ref{lem:weighted_composite_graphs_identical_kernels}]
Let $G$ be a weighted hybrid graph and $\mathcal{L}$ be its hybrid Laplacian matrix.
Its incidence matrix is $\mathcal{R} = \bm{[}B_l \quad R_q\bm{]}$, where $B_l$ and $R_q$ are the signed and the signed Laplacian matrices of its $L$- and $Q$-subgraphs, respectively.
Since $\mathcal{L} = \mathcal{R}\mathcal{R}^T$, by the same arguments as in the proof of Lemma \ref{lem:null_space_of_weighted_unweighted_graph_Laplacian_conincice}, the solutions to the equation $\mathcal{L} \mathbf{v} = 0$ are independent of the edge weights.
\end{proof}

\section{Construction of Bound Entangled States}
\label{apdx:bound_entangled_states}

The proofs of two results related to construction of bound entangled states are given here.

\degreecriterionsufficient*


\begin{proof}[Proof of Observation \ref{obs:degree_criterion_sufficient}]
Normalization is ignored as it has no effect on the definiteness of a matrix. Let $G$ be a $Q$-graph and $G^\Gamma$ be its partial transpose. Given $D(G) = D(G^\Gamma)$,

\begin{align}
    D(G) &= Q(G) - A(G) = D(G^\Gamma). \\
    \intertext{Thus,}
     Q(G) &= D(G^\Gamma) + A(G). \nonumber \\
    \implies Q^\Gamma(G) &= D^\Gamma(G^\Gamma) + A^\Gamma(G) \nonumber \\
    &= D(G^\Gamma) + A(G^\Gamma) \nonumber \\
    &= Q(G^\Gamma) \geq 0.
\end{align}
The same arguments apply to weighted and to hybrid graphs.
\end{proof}

\crosshatchgraphsareentangled*


\begin{proof}[Proof of Theorem \ref{thm:cross_hatch_graphs_are_entangled}]

An $m \times n$ cross-hatch $L$-graph is entangled for all $m, n \geq 3$ \cite{jh_phdthesis}.
Graph surgery procedures on $Q$- and $L$-graphs only differ in the STITCH step, which is not required for graph surgery on cross-hatch graphs, because connected components in cross-hatch graphs are either isolated vertices or single edges.
Therefore, the proof for $L$-graphs is sufficient for $Q$-graphs.

By Lemma \ref{lem:null_space_of_weighted_unweighted_graph_Laplacian_conincice}, weighted cross-hatch $L$- and $Q$-graphs are entangled.
Since graph surgery on hybrid graphs is based on $L$- and $Q$-surgeries, hybrid cross-hatch  graphs are entangled.
All cross-hatch graphs satisfy the degree criterion.
Thus they are bound entangled.
\end{proof}

\section{Hypergraphs}
\label{apdx:hypergraphs}

The degree criterion for hypergraph grid-states is proved below.

\hypergraphdegreecriterion*


\begin{proof}[Proof of Theorem \ref{thm:hypergraph_degree_criteion}]
Let $H$ be a hypergraph on $n$ vertices and $G$ be its graph. From Equation (\ref{eq:hypergraph-laplacian-with-offset-matrix})
\[
L(H) = Q(G) - O(H),
\]
where $O(H)$ is the offset matrix.
Then
\begin{align}
    L^\Gamma(H) &= Q^\Gamma(G) - O^\Gamma(H)  \nonumber \\
    &= Q(G^\Gamma) + \Delta - O(H),
\end{align}
where $\Delta = D(G) - D(G^\Gamma)$, and the second equality follows from applying Equation (\ref{eq:Q_degree_criterion_transposition}) to $G$.

The offset matrix $O(H)$ is positive semidefinite because it is a real, diagonal matrix with non-negative diagonal entries.
And from the PPT criterion, $L^\Gamma(H) \geq 0$, because $H$ represents a separable state.
This means
\begin{equation}
L^\Gamma(H) + O(H) = Q(G^\Gamma) + \Delta  \geq 0.
\end{equation}
Since $G^\Gamma$ is bipartite, from Lemma \ref{lem:singless_Laplacian_null_space_only_if_bipartite} and Lemma \ref{lem:null_space_of_weighted_unweighted_graph_Laplacian_conincice}, there exists a vector $\mathbf{v} \in \{-1, 1\}^n$ in $K[Q(G^\Gamma)]$. The matrix $\Delta$ is traceless and diagonal matrix. Thus, from Lemma \ref{lem:degree_criterion_support_lemma}, the matrix
\(
Q(G^\Gamma) + \Delta \ngeq 0.
\)
This is a contradiction. 
Therefore,  $\Delta = D(G) - D(G^\Gamma) = 0$.
\end{proof}

\onecolumngrid
\bibliographystyle{apsrev4-2}
\bibliography{ref}


\end{document}